\documentclass[journal]{IEEEtran}

\usepackage{amsmath,amssymb,amsthm,graphicx,mathrsfs,url,mathtools}
\usepackage[usenames,dvipsnames]{color}
\usepackage[hidelinks]{hyperref}
\usepackage{amsxtra}
\usepackage{epstopdf}
\usepackage{csquotes}

\usepackage{physics}

\usepackage{dsfont}

\usepackage{comment}

\usepackage[dvipsnames]{xcolor}

\usepackage{ifthen}

\newtheorem{theorem}{Theorem}
\newtheorem{lemma}{Lemma}

\newtheorem{corollary}{Corollary}

\theoremstyle{definition}

\newtheorem{remark}{Remark}

\newtheorem*{theorem*}{Theorem}
\newtheorem*{corollary*}{Corollary}

\DeclareMathOperator{\C}{\mathbb{C}}

\DeclareMathOperator{\cB}{\mathcal{B}}

\DeclareMathOperator{\cD}{\mathcal{D}}

\DeclareMathOperator{\cH}{\mathcal{H}}

\DeclareMathOperator{\cS}{\mathcal{S}}
\DeclareMathOperator{\cT}{\mathcal{T}}

\DeclareMathOperator{\1}{\mathds{1}}

\DeclareMathOperator{\supp}{supp}

\newcommand{\eps}{\varepsilon}
\newcommand{\II}{\mathds{1}}

\begin{document}

\title{Continuity bounds for quantum entropies arising from a fundamental entropic inequality}

\author{Koenraad~Audenaert,~Bjarne~Bergh,~Nilanjana~Datta,~Michael~G~Jabbour,~{\'A}ngela~Capel~and~Paul~Gondolf
\thanks{K. Audenaert is in the Institut f\"ur Theoretische Physik, Albert-Einstein-Allee 11, Universit\"at Ulm 89069, Germany (e-mail: koenraad.audenaert@uni-ulm.de).

B. Bergh is in the Department of Applied Mathematics and Theoretical Physics, Centre for Mathematical Sciences, University of Cambridge, Cambridge CB3 0WA, United Kingdom (e-mail: bb536@cam.ac.uk).

N. Datta is in the Department of Applied Mathematics and Theoretical Physics, Centre for Mathematical Sciences, University of Cambridge, Cambridge CB3 0WA, United Kingdom (e-mail: n.datta@damtp.cam.ac.uk).

M. G. Jabbour is at SAMOVAR, T\'el\'ecom SudParis, Institut Polytechnique de Paris, 91120 Palaiseau, France, and in the Centre for Quantum Information and Communication, \'Ecole polytechnique de Bruxelles, CP 165/59, Universit\'e libre de Bruxelles, 1050 Brussels, Belgium (e-mail: mjabbour@telecom-sudparis.eu).

{\'A}. Capel is in the Department of Applied Mathematics and Theoretical Physics, Centre for Mathematical Sciences, University of Cambridge, Cambridge CB3 0WA, United Kingdom, and at Fachbereich Mathematik, Universit\"{a}t T\"{u}bingen, 72076 T\"{u}bingen, Germany (e-mail: ac2722@cam.ac.uk).

P. Gondolf is at Fachbereich Mathematik, Universit\"{a}t T\"{u}bingen, 72076 T\"{u}bingen, Germany (e-mail: paul.gondolf@uni-tuebingen.de).}}

\maketitle

\begin{abstract}
    We establish a tight upper bound for the difference in von Neumann entropies between two quantum states, $\rho_1$ and $\rho_2$. This bound is expressed in terms of the von Neumann entropies of the mutually orthogonal states derived from the Jordan-Hahn decomposition of the difference operator $(\rho_1 - \rho_2)$. This yields a novel entropic inequality that implies the well-known Audenaert-Fannes (AF) inequality. In fact, it also leads to a refinement of the AF inequality.
    We employ this inequality to obtain a uniform continuity bound for the quantum conditional entropy of two states whose marginals on the conditioning system coincide.
    We additionally use it to derive a continuity bound for the quantum relative entropy in both variables. 
    Interestingly, the fundamental entropic inequality is also valid in infinite dimensions. 
\end{abstract}

\IEEEpeerreviewmaketitle

\section{Introduction}
Entropies play a crucial role in both classical and quantum information theory since they characterize optimal rates of various information-processing tasks. For example, for a discrete memoryless classical information source, its optimal rate of asymptotically reliable compression (i.e., its data compression limit) is given by its Shannon entropy~\cite{Shannon.1948}. For the case of a quantum information source, in an analogous setting, the data compression limit is given by its von Neumann entropy~\cite{VonNeumann.2013}. For a bipartite pure state, the von Neumann entropy of one of its marginals can also be used to quantify entanglement.

 There are other entropic quantities corresponding to bipartite systems, e.g.~the conditional entropy and the mutual information. The quantum (Umegaki) relative entropy and the Kullback-Leibler divergences act as parent quantities for all these entropies in the quantum and classical setting, respectively. Studying mathematical properties of all these quantities (which are also often referred to as {\em{information measures}}) has been the focus of much research. An important property of these quantities, which is of relevance in the study of various information-processing tasks, is that of {\em{continuity}}. For any such entropic quantity denoted by $f$, this property pertains to the following question: {\em{Given two quantum states, $\rho_1$ and $\rho_2$, that are close to each other with respect to a chosen distance measure, say $t$ (e.g.~the trace distance), how close is $f(\rho_1)$ to $f(\rho_2)$?}} In other words, it amounts to finding estimates of 
\begin{equation*}
    \operatorname{sup} \{ |f(\rho_1) - f (\rho_2)| \, : \,  t(\rho_1 , \rho_2) \leq \eps  \} \, .
\end{equation*}
A well-known continuity bound for the von Neumann entropy, $S(\rho) := - \Tr(\rho \log \rho)$, of a quantum state $\rho$, with respect to the trace distance, is referred to as the Audenaert-Fannes (AF) inequality~\eqref{AF} \cite{Fannes.1973,Petz.2008,Audenaert.2007}: For two quantum states $\rho_1, \rho_2$ (i.e.~positive semi-definite operators of unit trace) on a finite-dimensional Hilbert space, $\cH$,  with dimension $d$ that are $\eps$ close in trace distance, i.e.~$\frac{1}{2}||\rho_1 - \rho_2||_1 = \eps$, for some $\eps \in [0, 1]$, it holds that
\begin{align}\label{AF}
    |S(\rho_1) - S(\rho_2)| &\leq  \eps \log(d-1) + h(\eps) \, .
\end{align}

Similarly, this question was also studied for the conditional entropy, which is given by $S(A|B)_\rho = S(\rho_{AB}) - S(\rho_B)$ for a bipartite state $\rho_{AB}$,  
with $\rho_B$ being the marginal on the system $B$. Alicki and Fannes derived the first continuity bounds for this quantity in \cite{Alicki.2004}, with a later improvement by Winter in \cite{Winter.2016},  strengthening it to 
\begin{equation}
    |S(A|B)_{\rho_1} - S(A|B)_{\rho_2}| \le 2 \log d_A + (1 + \eps) h\Big(\frac{\eps}{1 + \eps}\Big) \, ,
\end{equation}
where $h$ denotes the binary entropy.
The importance of these results resides not only in their numerous applications but in the generality of the method employed to derive them, which is universal for multiple entropic quantities. This method was coined \textit{AFW method} by Shirokov in \cite{Shirokov.2020, Shirokov.2022} after the original authors, and subsequently named \textit{ALAFF method} (for \enquote{Almost Locally AFFine}) in \cite{Bluhm.2023.1, Bluhm.2023.3} due to the main property exploited in it. In the past few years, it has been multiply used not only to derive better continuity bounds for quantities derived from the Umegaki relative entropy in finite \cite{Synak.2006,Gour.2021,Rastegin.2011,Reeb.2015} and infinite dimensions \cite{Shirokov.2019,Shirokov.2018,Shirokov.2022.2}, but also for other quantities such as Rényi divergences \cite{Mosonyi.2011,Bluhm.2023.2}, the Belavkin-Staszewski relative entropy \cite{Bluhm.2023.1}, the fidelity \cite{Gilyen.2022}, and more.

In this paper, we introduce a new upper bound on $S(\rho_1) - S(\rho_2)$ which will turn out to imply the AF inequality \eqref{AF}, and also lead to a uniform continuity bound for the quantum conditional entropy when the marginals on the conditioning system agree, and to a continuity bound on the quantum (Umegaki) relative entropy. Let us introduce the inequality in various equivalent forms. 
Consider the Jordan-Hahn decomposition of the difference $\rho_1 - \rho_2$:
\begin{equation*}
     \rho_1 - \rho_2 = \Delta_+ - \Delta_- \, ,
\end{equation*}
 where $\Delta_\pm$ are positive semi-definite operators of orthogonal supports (which we write as $\Delta_+ \perp \Delta_-$). We can express this difference in terms of quantum states $\rho_\pm$, namely, $\rho_1 - \rho_2 = \eps \rho_+ - \eps\rho_-$, with the states $\rho_\pm$ being defined through the relation:  $\Delta_\pm = \eps \rho_\pm$ if $\eps > 0$, and for $\eps = 0$ we can make an arbitrary choice of states $\rho_+$ and $\rho_-$ such that $\rho_+ \perp \rho_-$. In this paper, we prove that the difference of the von Neumann entropies of the states $\rho_1$ and $\rho_2$ can be expressed in terms of the entropies of the states $\rho_\pm$ via the following inequality:
\begin{align}\label{Inineq1}
    S(\rho_1) - S(\rho_2) \leq \eps S(\rho_+) - \eps S(\rho_-) + h(\eps)\,,
\end{align}
where $h(\eps) := - \eps \log \eps - (1-\eps) \log (1-\eps)$ denotes the binary entropy.
Moreover, this bound is {\em{tight}}, in the sense that, for any $\eps \in [0,1]$, there exist pairs of states for which the bound is saturated.
This inequality can be cast in various equivalent forms. Firstly, note that interchanging $\rho_1$ and $\rho_2$ results in interchanging $\Delta_+$ and $\Delta_-$. Hence, we also have
\begin{align}\label{Inineq2}
    S(\rho_2) - S(\rho_1) \leq \eps S(\rho_-) - \eps S(\rho_+) + h(\eps)\,.
\end{align}
As an immediate consequence of \eqref{Inineq1} and \eqref{Inineq2} one gets
\begin{align}\label{Inineq3}
    |\left(S(\rho_1) - S(\rho_2)\right) - \eps \left(S(\rho_+) - \eps S(\rho_-)\right)| \leq h(\eps) \, .
\end{align}
We see that (\ref{Inineq1}) is more fundamental than the AF inequality~\eqref{AF}.  This is because from (\ref{Inineq1}) it follows that (assuming without loss of generality that $S(\rho_1) \le S(\rho_2)$),
\begin{equation}
    \begin{aligned}
        |S(\rho_1) - S(\rho_2)| &\leq  \eps S(\rho_+) - \eps S(\rho_-) + h(\eps)\\
        &\leq \eps S(\rho_+)+ h(\eps)\\
        &\leq  \eps \log(d-1) + h(\eps) \, ,
    \end{aligned}
\end{equation}
which is the right-hand side of \eqref{AF}. In the second line, we have used the non-negativity of the von Neumann entropy, and the last line follows from the fact that $\rank \, \rho_+ \leq d - 1$. 

To gain some intuition behind our new entropic inequality (considered in any of its equivalent forms), let us first consider some simple cases where it can be easily seen to hold.   
\smallskip

\noindent
{\em{Case 1:}} The inequality clearly holds when $\rho_1$ and $\rho_2$ are states of qubits, i.e.~when $\cH \simeq \C^2$. This is because in this case, $\rho_\pm$ are pure states and hence $S(\rho_\pm) =0$.
Therefore, (\ref{Inineq3}) reduces to 
\begin{equation*}
    |S(\rho_1) - S(\rho_2)| \leq h(\eps) \, ,
\end{equation*}
which is just the AF inequality (\ref{AF}) for the case of qubits (i.e.~$d=2$).
\smallskip

\noindent
{\em{Case 2:}} The inequality holds whenever $\rho_1 \geq \eps \rho_+$, which in turn guarantees that $\rho_2 \geq \eps \rho_-$, since via the Jordan-Hahn decomposition we have that $\rho_1 - \eps \rho_+ = \rho_2 - \eps \rho_-.$
The inequality can then be proved as follows. Note that $\Tr (\rho_1 - \eps \rho_+) = 1- \eps = \Tr (\rho_2 - \eps \rho_-)$. Note first that if $\eps = 1$, then $\rho_1 \perp \rho_2$ and so $\rho_+ = \rho_1$, $\rho_- = \rho_2$ and so the inequality holds trivially. If $\eps < 1$, then let us define the quantum state
\begin{align}
    \omega := \frac{\rho_1 - \eps \rho_+}{1 - \eps} \equiv \frac{\rho_2 - \eps \rho_-}{1 - \eps} \, .
\end{align}
Then, we can write convex decompositions of the states $\rho_1$ and $\rho_2$ as follows:
\begin{align}   
    \rho_1 &= \eps \rho_+ + (1-\eps) \omega \, ,\label{s1}\\
    \rho_2 &= \eps \rho_- +  (1-\eps) \omega \, .\label{s2}
\end{align}
The property of ``almost convexity'' of the von Neumann entropy \cite{Kim.2013,Kim.2014} applied to (\ref{s1}) implies that
\begin{align}\label{r1}
    S(\rho_1) \leq \eps S(\rho_+) + (1-\eps) S(\omega) + h(\eps) \, . 
\end{align}
Moreover, the concavity of the von Neumann entropy applied to (\ref{s2}) implies that
\begin{align}\label{r2}
    S(\rho_2) \geq \eps S(\rho_-) +  (1-\eps) S(\omega) \, . 
\end{align}
Then from (\ref{r1}) and (\ref{r2}) we immediately obtain the desired inequality (\ref{Inineq1}).
\begin{align}
     S(\rho_1) - S(\rho_2) \leq \eps S(\rho_+) - \eps S(\rho_-) + h(\eps) \, .
\end{align}
\smallskip

\noindent
\begin{remark}
    It can also be easily seen that the inequality holds when $\rho_1$ and $\rho_2$ commute, as this is a special case of {\em{Case 2}}. Note first that if $\rho_1$ and $\rho_2$ commute then $\rho_1$ and $\rho_2$ also commute with $\rho_\pm$. Also,  the states $\rho_1$ and $\rho_2$ are then simultaneously diagonalizable and hence the operator inequality $\rho_1 \geq \eps \rho_+$ reduces to an inequality between eigenvalues of $\rho_1$ and $\rho_+$. Let us fix some ordering of the vectors in this simultaneous eigenbasis, and then write the eigenvalue corresponding to the $i^{th}$ basis vector as $\lambda_i(\sigma)$ for any state $\sigma$ that is diagonal in this basis. Then the operator inequality $\rho_1 \geq \eps \rho_+$ reduces to
    \begin{align} 
    \lambda_i(\rho_1) \geq \eps  \lambda_i(\rho_+) \quad \forall \,\, i \in [d],
    \end{align}
    where $[d]$ denotes the index set of $d$ elements. Let $p_i :=  \lambda_i(\rho_1)$ and $q_i :=  \lambda_i(\rho_2),$ for all $i \in [d]$, and define the sets
    \begin{align}
        \begin{aligned}
            I:= \{ i \in [d]\,:\, p_i \geq q_i\} \, , \\
            I^c:= \{ i \in [d]\,:\, p_i < q_i\} \, .
        \end{aligned}
    \end{align}
    Then $\lambda_i(\rho_1 - \eps \rho_+) = \lambda_i(\rho_1) - \eps \lambda_i(\rho_+),$ where $\lambda_i(\rho_+) = p_i - q_i$ for all $i \in I$ and is equal to zero else. Hence, we have that for all $i \in I,$ $\lambda_i(\rho_1 - \eps \rho_+) = (1-\eps) p_i + \eps q_i \geq 0 $ and for all $i \in I^c$, $\lambda_i(\rho_1 - \eps \rho_+) = p_i \geq 0$. Thus the required inequality of {{Case 2}}, namely, $\rho_1 \geq \eps \rho_+$, holds in this case.
\end{remark}
\smallskip

\noindent
{\bf{Layout of the paper:}} 
Our main result, namely the above-mentioned entropic inequality, is stated in Theorem~\ref{Thm1} of Section~\ref{sec:main}, and a sharper version of it is stated in Theorem~\ref{Thm-sharp}. An extension of this inequality to conditional entropies is given in Theorem~\ref{thm:curious-condl}.  In Section~\ref{sec:key-lemmas} we state and prove a few key lemmas that we employ in the proof of the above theorems and of subsequent results. The proof of Theorem~\ref{Thm-sharp} is given in Section~\ref{sec:proofThm1}. In Section~\ref{sec:ref-AF} we use our fundamental entropic inequality to state and prove a refined version of the AF inequality~(\ref{AF}); see Theorem~\ref{thm-ref-AF}. In Section~\ref{sec:applications}, we apply Theorem~\ref{Thm1} to obtain a uniform continuity bound for the conditional entropy whenever the marginals on the second system agree, and a continuity bound on the quantum relative entropy. These are stated in Theorem~\ref{thm:continuity-conditional-entropy} and  Corollary~\ref{cor:RelEntbothinputs}, respectively, and their proofs are presented in the same section. We end the paper with an extension of our fundamental entropic inequality from Theorem~\ref{Thm1} to infinite dimensions in Section \ref{sec:infinite}.

\section{Main Results}
\label{sec:main}

For the majority of this paper, with the exception of Section \ref{sec:infinite} at the end where we deal with infinite-dimensional Hilbert spaces, we restrict attention to a finite-dimensional Hilbert space $\cH$ of dimension $d$. Let $\cB(\cH)$ denote the algebra of linear operators acting on $\cH$, and $\cB_{\operatorname{sa}}(\cH)$ denote the subset of self-adjoint ones. The set of quantum states (density matrices), i.e.~positive semi-definite operators of unit trace is denoted by $\cD(\cH) \subset \cB_{\operatorname{sa}}(\cH)$, and its subset of positive definite operators of unit trace by $\cD_+(\cH) $. For $X \in \cB(\cH)$, we denote the kernel of $X$ as $\ker X = \{\ket{\psi} \in \cH \;:\; X\ket{\psi} = 0\}$ and its support by $\supp X = (\ker X)^\perp$. Note that when writing $A \le B$ for $A, B \in \cB_{\operatorname{sa}}(\cH)$ we refer to the Loewner partial order. The norms on $\cB(\cH)$ that we use are the trace- or one-norm $\norm{\,\cdot\,}_1$ and the operator- or infinity-norm $\norm{\,\cdot\,}_\infty$ both of which are members of the wider family of Schatten-p-norms $\norm{A}_p = \Tr((A^*A)^{p/2})^{1/p}$ for $p \in [1, \infty)$ (where $\norm{\,\cdot\,}_\infty$ corresponds to the limit $p \to \infty$). The trace distance between two density matrices $\rho_1, \rho_2 \in \cD(\cH)$ is given by $\frac{1}{2}\norm{\rho_1 - \rho_2}_1$.

For any vector $\underline{p} \in \mathbb{R}^d$ of non-negative entries (not necessarily a normalized vector), we define its Shannon entropy $H(\underline{p})$ as $H(\underline{p}) := \sum_i \eta(p_i) := - \sum_i p_i \log p_i$. Additionally, for any positive semi-definite operator $
A \in \cB_{\operatorname{sa}}(\cH)$, $A \geq 0$, we define its von Neumann entropy as $S(A) := - \Tr(A \log A)$.  The quantum (Umegaki) relative entropy of a state $\rho$ with respect to a positive semi-definite operator $A$ is given by \begin{align}
    D(\rho \Vert A) &= \begin{cases}
        \Tr(\rho \log \rho - \rho \log A) & \text{if } \ker A \subseteq \ker \rho,\\
        \infty & \text{else}. 
    \end{cases}
\end{align}
For a bipartite state $\rho_{AB} \in \cD(\cH_A \otimes \cH_B)$, the conditional entropy of the system $A$ given the system $B$ is given by $S(A|B)_\rho:= S(\rho_{AB}) - S(\rho_B)$, where $\rho_B= \Tr_A \rho_{AB}$ denotes the reduced state (i.e.~marginal) of the system $B$. It can be expressed in terms of a relative entropy as follows:
\begin{align}
\label{condrel}
    S(A|B)_\rho & = - D(\rho_{AB}\Vert \1_A \otimes \rho_B) \nonumber \\
    & = \max_{\nu_{B} \in \cD(\cH_B)} \left[-D(\rho_{AB} \| \1_A \otimes \nu_{B})\right],
\end{align}
where $\1_A$ denotes the identity operator on the system $A$.
 We also employ the max-relative entropy~\cite{Datta.2009} which is defined as follows\footnote{The definition in the original paper~\cite{Datta.2009}  is  $D_{\max}(\rho \| \sigma):= \inf \{ \lambda >0 \, : \, \rho \leq 2^\lambda \sigma \} $. However, we consider here a slightly modified version since we are using natural logarithms throughout the whole text.}: 
 \begin{equation}
     D_{\max}(\rho \| \sigma):= \inf \{ \lambda >0 \, : \, \rho \leq e^\lambda \sigma \} . \label{dmax}
 \end{equation}
 Note that throughout this paper, we use $\log$ to denote the natural logarithm.

We are now in position to state our main results.

\begin{theorem}
\label{Thm1}
    Let $\rho_1, \rho_2 \in \cal{D}(\cal{H})$, with $\dim \cH = d$, such that $\frac{1}{2}\norm{\rho_1 - \rho_2}_1 = \eps$, for some $\eps \in [0,1]$. 
    Let $\rho_{\pm}$ be the normalized Jordan-Hahn decomposition of $(\rho_1-\rho_2)$ as described above, i.e. 
\begin{align}
        \rho_1 - \rho_2 &= \eps \rho_+ - \eps\rho_- \, ,
        \label{two}    \end{align}
        where $\rho_\pm \in \cal{D}(\cal{H}_A \otimes \cal{H}_B)$ and $\rho_+ \perp \rho_-$.
    Then
    \begin{align}
        S(\rho_1) - S(\rho_2) \leq \eps S(\rho_+) - \eps S(\rho_-) + h(\eps) \,.
        \label{ineq1}
    \end{align}
    Moreover, the inequality is tight.
\end{theorem}
\begin{remark}
     To see that the bound~\eqref{ineq1} is tight, for every $\eps \in [0,1]$ one can simply consider the following commuting states $\rho_1$ and $\rho_2$:
\begin{equation}
    \rho_1 = (1-\eps) \ketbra{\psi} + \frac{\eps}{d-1} \left( \mathds{1} - \ketbra{\psi} \right) \quad \text{and} \quad \rho_2 = \ketbra{\psi} \, ,
\end{equation}
where $\ket{\psi}$ is any pure state, while $\mathds{1}$ denotes the identity operator in $\cB(\cH)$.
\end{remark}
\bigskip

\noindent We prove this theorem in the next section. The above theorem extends to conditional entropies for bipartite states if the condition given in (\ref{condition}) below holds. This result is stated in Theorem~\ref{thm:curious-condl}.
\begin{theorem}
\label{thm:curious-condl}
    Let $\rho_1, \rho_2 \in \cal{D}(\cal{H}_A \otimes \cal{H}_B)$ and $\frac{1}{2}\norm{\rho_1 - \rho_2}_1 = \eps$, for some $\eps \in [0,1]$. 
    Let $\rho_{\pm}$ be the normalized Jordan-Hahn decomposition of $(\rho_1-\rho_2)$ as described above, i.e. 
\begin{align}
        \rho_1 - \rho_2 &= \eps \rho_+ - \eps\rho_- \, ,
        \label{two2}    \end{align}
        where $\rho_\pm \in \cal{D}(\cal{H}_A \otimes \cal{H}_B)$ and $\rho_+ \perp \rho_-$. Further, assume that the operator $K$ defined below is positive semi-definite, i.e.~
    \begin{align}
    \label{condition}
        K:= \rho_1- \eps \rho_+ = \rho_2 - \eps \rho_- \geq 0.
    \end{align}
    Then
    \begin{align}
        S(A|B)_{\rho_1} - S(A|B)_{\rho_2}  \leq \eps S(A|B)_{\rho_+} - \eps S(A|B)_{\rho_-} + h(\eps).
        \label{cond-ineq1}
    \end{align}
\end{theorem}
\begin{proof}
Similarly to above, if $\eps = 1$, then $\rho_1 = \rho_+$ and $\rho_2 = \rho_-$ and so the relation holds trivially. For $\eps < 1$, let $\omega \in \cal{D}(\cal{H}_A \otimes \cal{H}_B)$ be defined through the relation
$K = (1 - \eps) \omega$. Then 
\begin{align}
    \rho_1 &= \eps \rho_+ + (1-\eps) \omega \label{eqq1}\\
     \rho_2 &= \eps \rho_- + (1-\eps) \omega \label{eqq2}.
\end{align}
Note that the conditional entropy $S(A|B)_{\rho_1}$ is given by
\begin{align}
   S(A|B)_{\rho_1} 
   &= S(\rho_1) - S(\rho_{1, B}) \nonumber \\
   &= S(\rho_1) + \Tr \left( \rho_1 \log (\1_A \otimes \rho_{1,B})\right),
   \label{eqq3}
\end{align}
where $\rho_{1,B} = \Tr_A \rho_1$. Further, by (\ref{eqq1}), 
\begin{align}
  \Tr \left( \rho_1 \log (\1_A \otimes \rho_{1,B})\right) & = \eps \Tr \left( \rho_+ \log (\1_A \otimes \rho_{1,B})\right) \nonumber \\
  & \quad + (1- \eps) \Tr \left( \omega \log (\1_A \otimes \rho_{1,B})\right)  
\end{align}
and
\begin{align}
    S(\rho_1) &\le \eps S(\rho_+) + (1-\eps) S(\omega) + h(\eps),
\end{align}
where the last inequality follows from the property of ``almost convexity" of the von Neumann entropy. The above inequalities imply that
\begin{align}
    S(A|B)_{\rho_1} & = S(\rho_1) + \Tr \left( \rho_1 \log (\1_A \otimes \rho_{1,B})\right) \nonumber\\
    &\leq h(\eps) - \eps D(\rho_+ \Vert \1_A \otimes \rho_{1,B}) \nonumber \\
    & \quad - (1- \eps) D(\omega \Vert \1_A \otimes \rho_{1,B})\nonumber\\
    &\leq h(\eps) + \eps S(A|B)_{\rho_+} + (1-\eps) S(A|B)_\omega ,
    \label{eqq5}
\end{align}
where in the last step we also used the variational characterization of the conditional entropy, i.e. \eqref{condrel}.
On the other hand, by (\ref{eqq2}) and the concavity of the conditional entropy, we have
\begin{align}
    S(A|B)_{\rho_2} & \geq \eps S(A|B)_{\rho_-} + (1-\eps) S(A|B)_\omega.
    \label{eqq6}
\end{align}
Inequalities (\ref{eqq5}) and (\ref{eqq6}) yield the desired inequality~(\ref{cond-ineq1}). 
\end{proof}

\section{Two key lemmas}
\label{sec:key-lemmas}


The proof of Theorem~\ref{Thm1} (and its sharper version, Theorem~\ref{Thm-sharp} stated in the next section) will be simplified if we extend the definition of the von Neumann entropy to positive operators that
do not necessarily have trace 1. Remember that we defined the functional
\begin{align}
\label{vN1}
S(A) := -\Tr (A \log A)
\end{align}
for every positive operator $A$, $A\ge0$.
Clearly, when $\rho$ is a state, $a\ge0$, and $A=a\rho$, we have
\begin{align}
S(A) = a S(\rho) - a \log a
\label{vN2}
\end{align}
This identity allows to generalise the inequalities (\ref{r1}) and (\ref{r2}) expressing almost convexity and concavity of the von Neumann entropy, respectively,
to this entropy functional. One easily obtains the following, for $A,B\ge0$ with $a=\Tr A$, $b=\Tr B$:
\begin{eqnarray}
S(A+B) &\le& S(A) + S(B) \label{eq:subadd}\\
S(A+B) &\ge& S(A)+S(B) -(a+b) h\left(\dfrac{b}{a+b}\right).
\label{eq:almost-super}\end{eqnarray}
Thus, almost convexity turns into functional subadditivity (not to be confused with the usual subadditivity of the von Neumann entropy with respect
to addition of subsystems), and concavity turns into functional almost-super-additivity.

We now show that the latter inequality can be made sharper when $B$ is not full rank.
That this should be possible is already being hinted at by the existence of the identity $S(A+B)=S(A)+S(B)$ when $A$ and $B$
have orthogonal supports. The following lemma extends this fact.
\begin{lemma}[Sharpened almost-superadditivity]\label{SASA}
Let $A,B\ge 0$ and ${\mathcal M}=\supp B$. Denoting the restriction of an operator $X$ to $\cal M$ by $X|_{\mathcal M}$, and
defining
$a' = \Tr A|_{\mathcal M}$ and $b= \Tr B \equiv  \Tr B|_{\mathcal M}$,
we have
\begin{eqnarray}
S(A+B) - S(A) 
&\ge& S((A+B)|_{\mathcal M}) - S(A|_{\mathcal M}) \nonumber\\
&\ge& S(B) - (a'+b) h\left(\dfrac{b}{a'+b}\right)\label{eq:SASA}
\end{eqnarray}
\end{lemma}
This inequality will be an essential ingredient in the proof of Theorem~\ref{Thm1}.

\medskip

\noindent
\begin{proof}

Monotonicity of the Holevo chi $\chi$ under a CPTP map $\Phi$, applied to a two-element ensemble, explicitly reads as follows:
\begin{align}
    & S(p\rho + (1-p)\sigma) - p S(\rho) - (1-p) S(\sigma) \nonumber \\
    & \ge S(p\Phi(\rho) + (1-p)\Phi(\sigma)) - p S(\Phi(\rho)) - (1-p) S(\Phi(\sigma)),
\end{align}
or, rephrased in terms of positive operators $A$ and $B$,
\begin{align}
    S(A+B) - S(A)-S(B) & \ge S(\Phi(A)+\Phi(B)) \nonumber \\
    & \quad - S(\Phi(A)) - S(\Phi(B)).
\end{align}
The third term of each side drops out if $\Phi$ leaves $B$ unchanged.
Let, in particular, $\Phi$ be a pinching to the subspaces $\mathcal{M}=\supp B$ and $\mathcal{M}^\perp=\ker B$.
Then
\begin{eqnarray*}
S(A+B)-S(A) 
&\ge& S((A+B)|_{\mathcal M} \oplus (A+B)|_{\mathcal{M}^\perp}) \\
&\phantom{=}& - S(A|_{\mathcal M} \oplus A|_{\mathcal{M}^\perp}) \\
&=& S((A+B)|_{\mathcal M}) - S(A|_{\mathcal M}),
\end{eqnarray*}
which is the first inequality of the lemma.

The second inequality of the lemma then follows by exploiting almost super-additivity of $S$ given by~\eqref{eq:almost-super}.
\end{proof}

\begin{lemma}
\label{lem:auxiliary-estimate}
    For $\rho, \sigma \in \cD(\cH)$ with $\rho \perp \sigma$ (i.e.~they have mutually orthogonal supports) and $\omega = t \rho + (1 - t) \nu$ with $t \in (0, 1)$ and $\nu \in \cD(\cH)$, one has
    \begin{equation}
        D(\rho \Vert \omega) - D(\sigma \Vert \omega) \le \log(\frac{1}{t} - 1)  .
    \end{equation}
\end{lemma}
\begin{proof}
    From $\omega = t \rho + (1 - t) \nu$ it follows that $ \omega \ge t \rho$ which in turn gives
    \begin{equation}
        D(\rho \Vert \omega) \le D(\rho \Vert t \rho) = -\log t \, .
    \label{t23}
    \end{equation}
    Let us define the pinching map $\Phi$ which acts on any $\tau \in \cD(\cH)$ as follows: 
    \begin{equation}
       \Phi(\tau) := 
       P_{\sigma} \tau P_{\sigma} + P_{\sigma}^\perp \tau P_{\sigma}^\perp ,
    \end{equation}
 where $P_{\sigma}$ and $P_{\sigma}^\perp$ denote orthogonal projections onto the support of $\sigma$ and its complement, respectively. 
 Then, by the data-processing inequality, we have
    \begin{equation}\label{eq:DPIrelativ-enetropy}
        \begin{aligned}
           D(\sigma \Vert \omega)
            & \geq  D(\Phi(\sigma) \Vert \Phi(\omega)) \\
            & =  D(\sigma \Vert t \rho + (1 - t) \Phi(\nu))  \\
 &  =  D(\sigma \Vert (1 - t) \Phi(\nu)) ,
        \end{aligned}
    \end{equation}
 where the last equality holds because $\rho \perp \sigma$. Therefore,
    \begin{equation}
        D(\sigma \Vert \omega) \ge  -\log(1 - t) + D(\sigma \Vert \Phi(\nu)) \ge -\log(1 - t), 
    \label{t24}
    \end{equation}
  due to the non-negativity of the relative entropy between two quantum states. The bounds (\ref{t23}) and (\ref{t24}) together yield the statement of the lemma, since
     \begin{equation}
       D(\rho \Vert \omega)- D(\sigma \Vert \omega) \le  -\log t  + \log(1 - t)  =  \log(\frac{1}{t} - 1)  .  
    \label{eq:diff_relative_entropies}
     \end{equation}
\end{proof}
\medskip

\noindent We are now in a position to prove Theorem~\ref{Thm1}.

\section{Proof of Theorem~\ref{Thm1}}
\label{sec:proofThm1}


As mentioned earlier, we actually prove an inequality which is sharper than the one stated in Theorem~\ref{Thm1}. It not only involves the quantity $\eps$ (i.e.~the trace distance between the states $\rho_1$ and $\rho_2$)
but also 
$$
c := \Tr \rho_2|_{\mathcal M}
$$
where ${\mathcal M}$ denotes the support of $\rho_-$. We are grateful to Peter Frenkel for inquiring about a possibility of such kind. The sharper inequality is stated in the following theorem.
\bigskip

\begin{theorem}\label{Thm-sharp}
Let $\rho_1, \rho_2 \in \cal{D}(\cal{H})$, with $\dim \cH = d$, such that $\frac{1}{2}\norm{\rho_1 - \rho_2}_1 = \eps$, for some $\eps \in [0,1]$. 
    Let $\rho_{\pm}$ be the normalized Jordan-Hahn decomposition of $(\rho_1-\rho_2)$ as described above i.e. 
\begin{align}
        \rho_1 - \rho_2 &= \eps \rho_+ - \eps\rho_- \, \label{two3},
   \end{align}
        where $\rho_\pm \in \cal{D}(\cal{H}_A \otimes \cal{H}_B)$ and $\rho_+ \perp \rho_-$. Further, let $c := \Tr \rho_2|_{\mathcal M}$, where ${\mathcal M}$ denotes the support of $\rho_-$. Then
    \begin{align}
        S(\rho_1) - S(\rho_2) \leq \eps S(\rho_+) - \eps S(\rho_-) + c \, h\left(\frac{\eps}{c}\right). \, 
        \label{ineq1-sharp}
    \end{align}
\end{theorem}
To see that this inequality is indeed sharper than the one stated in Theorem~\ref{Thm1}, we need to establish that 
\begin{align}
 c \, h\left(\frac{\eps}{c}\right) \le h(\eps),   
\end{align}
This follows from concavity of the binary entropy $h$ and the fact that $0\le \eps\le c\le 1$ (which in turn follows from taking the trace of the restriction to $\mathcal{M}=\supp \rho_-$ of the identity (\ref{two3})):
\begin{align}
h(\eps) & = h\left( c\frac{\eps}{ c}+(1- c) 0\right) \nonumber \\
& \ge  c \, h\left(\frac{\eps}{ c}\right)+(1- c) h(0) =  c \, h\left(\frac{\eps}{ c}\right).\end{align}
\medskip

\noindent
We now proceed to prove Theorem~\ref{Thm-sharp}.
\begin{proof}
Let us define the positive operator $M$ by
$$
M := \rho_1+\Delta_- = \rho_2+\Delta_+.
$$
Then we have
\begin{align}
S(\rho_1) - S(\rho_2) 
& = S(\rho_1) - S(M) + S(M) - S(\rho_2) \\
& = -(S( \rho_1+\Delta_-)-S(\rho_1)) \nonumber \\
& \quad + (S(\rho_2+\Delta_+) - S(\rho_2)).
\end{align}
To find an upper bound on the first term we use the inequality (\ref{eq:SASA}) of Lemma~\ref{SASA} (sharpened almost super-additivity), and note
that $\Tr \Delta_- = \eps$ and $\Tr(\rho_1+\Delta_-)|_{\mathcal{M}} = \Tr\rho_2|_{\mathcal{M}} =c$.
For the second term we use subadditivity~(\ref{eq:subadd}). This gives
\begin{align}
S(\rho_1) - S(\rho_2) 
&\le -\left(S(\Delta_-) - c \, h\left(\frac{\eps}{c}\right) \right) + S(\Delta_+) \nonumber\\
& =  S(\Delta_+)-S(\Delta_-) + c \, h\left(\frac{\eps}{c}\right)\nonumber\\
&= \eps S(\rho_+)-\eps S(\rho_-) + c \, h\left(\frac{\eps}{c}\right).
\end{align}
  
\end{proof}

\section{A refined continuity bound for the von Neumann entropy}
\label{sec:ref-AF}
The fundamental entropic inequality stated in Theorem~\ref{Thm1} along with Lemma~\ref{lem:auxiliary-estimate} leads to the refinement of the AF inequality~(\ref{AF}) given by Theorem~\ref{thm-ref-AF} below. 
\begin{theorem}\label{thm-ref-AF}
 Let $\rho_1, \rho_2 \in \cal{D}(\cal{H})$, with $\dim \cH = d$, such that $\frac{1}{2}\norm{\rho_1 - \rho_2}_1 = \eps$, for some $\eps \in [0,1]$. Let
 $\rho_1 - \rho_2 = \eps \rho_+ - \eps\rho_-$ where $\rho_\pm \in \cal{D}(\cal{H})$, and $\rho_+ \perp \rho_-$. Then
    \begin{multline}
        |S(\rho_1) - S(\rho_2)|  \\ \leq \eps \log\big(d \max\{\lambda_{\max}(\rho_-), \lambda_{\max}(\rho_+)\} - 1\big) + h(\eps) \, .
        \label{ineq1p}
    \end{multline}       
\end{theorem}

\smallskip

\noindent
\begin{remark}
   Note that Berta {\em{et al}}~\cite{Berta.2024} proved an analogous result but with $\lambda_{\max}(\rho_-)$ replaced by $\lambda_{\max}(\rho_2)$, and for $\eps \le 1 - (1/(d \lambda_{\max}(\rho_2))$; see Corollary 3 of~\cite{Berta.2024}. In fact, it was their result that inspired us to prove Theorem~\ref{thm-ref-AF}.
\end{remark}

Before proving the above theorem, assume without loss of generality that $S(\rho_1) \geq S(\rho_2)$. Recall that equality holds in the AF inequality (\ref{AF}) for the following choice:
$$\rho_1 = (1-\eps) \ketbra{1} + \frac{\eps}{d-1} \sum_{i=2}^d \ketbra{i}, \quad {\hbox{and}} \quad \rho_2 = \ketbra{1},$$
since $S(\rho_2)=0$ and $S(\rho_1) = \eps \log (d-1) + h(\eps).$

Note that for this choice of $\rho_1$ and $\rho_2$, we have 
$$\rho_+ = \frac{1}{d-1} \sum_{i=2}^d \ketbra{i} \quad {\hbox{and}} \quad \rho_- = \ketbra{1}.$$
Hence, $\lambda_{\max}(\rho_-)=1$
so that the RHS of (\ref{ineq1p}) 
reduces to the RHS of the AF inequality (\ref{AF}), and the inequality is saturated. 

\subsection{Proof of Theorem~\ref{thm-ref-AF}}
\begin{proof}
By Theorem~\ref{Thm1},
\begin{align}
    S(\rho_1) - S(\rho_2) &\leq \eps (S(\rho_+) - S(\rho_-)) + h(\eps)\nonumber\\
    &= \eps (D(\rho_-\|\tau) - D(\rho_+\|\tau) ) + h(\eps),
    \label{stp0}
\end{align}
where $\tau= I/d,$ the completely mixed state. Let the spectral decomposition of $\rho_-$ be given by
\begin{align}
    \rho_- = \sum_{i=1}^d \lambda_i \ketbra{e_i} &\leq \lambda_{\max}(\rho_-) I= \frac{1}{t} \tau,
    \label{stp1}
\end{align}
where $t := \frac{1}{d\lambda_{\max}(\rho_-)}$. Setting $\omega \equiv \tau$, we obtain from (\ref{stp1})
\begin{align}
    \omega \equiv \tau = t \rho_- + (1-t) \nu,\label{stp2}
\end{align}
where $\nu = \frac{\tau - t \rho_-}{1-t} \in \cD(\cH).$ Since $\rho_+ \perp \rho_-$ and (\ref{stp2}) holds, we can apply Lemma~\ref{lem:auxiliary-estimate} to the RHS of (\ref{stp0}) to obtain
\begin{align}
    {\hbox{RHS of (\ref{stp0})}} &\leq \eps \log (\frac{1}{t} - 1) + h(\eps) \nonumber\\
    &=  \eps \log(d \lambda_{\max}(\rho_-) - 1) + h(\eps)\,.
    \label{stp3}
\end{align}
To obtain (\ref{ineq1p}), note when exchanging $\rho_1$ and $\rho_2$, $\rho_+$ will become $\rho_-$ and vice-versa, and so we can obtain the absolute value bound of (\ref{ineq1p}) by combining (\ref{stp3}) with its version where $\rho_1$ and $\rho_2$ have been exchanged.
\end{proof}

\section{Further Applications}\label{sec:applications}

As a straightforward application of our inequality (\ref{ineq1}), we obtain a uniform  continuity bound for the conditional entropy $S(A|B)_\rho$, in the special case in which the marginals of the two states on the system $B$ are identical. The following has been conjectured (also in the case where the two marginals are not identical)~\cite{Winter.2016, Wilde.2020}: For $\rho_1, \rho_2 \in \cD(\cH_A \otimes \cH_B)$, such that $\frac{1}{2} \Vert \rho_1 - \rho_2 \Vert_1 = \eps$,
\begin{equation}\label{eq:conjecture}
    |S(A|B)_{\rho_1} - S(A|B)_{\rho_2}| \le \eps \log(d_A^2 - 1) + h(\eps) \, .
\end{equation}
This continuity  bound is tight as the inequality is saturated for certain choices of states $\rho_1, \rho_2 \in \cD(\cH_A \otimes \cH_B)$   \cite{Wilde.2020,Alhejji.2023}. The analogous classical inequality has appeared e.g. in \cite{Zhang.2007} and \cite{Alhejji.2020}.
Our proof of~\eqref{eq:conjecture} for the special case of equal marginals is given by Theorem~\ref{thm:continuity-conditional-entropy}.

Using Lemma~\ref{lem:auxiliary-estimate}, we can easily prove the following result on the uniform continuity bound for the conditional entropy of bipartite states with equal marginals, using our fundamental inequality~(\ref{ineq1}) stated in Theorem~\ref{Thm1}.

\begin{theorem}\label{thm:continuity-conditional-entropy}
    For $\rho_1, \rho_2 \in \cD(\cH_A \otimes \cH_B)$ with equal marginals $\rho_{1, B} = \rho_{2, B}$ and $\frac{1}{2}\Vert \rho_1 - \rho_2\Vert_1 = \eps$ for some $\eps \in [0, 1]$, we find that
    \begin{equation}\label{eq:continuity-conditional-entropy}
        |S(A|B)_{\rho_1} - S(A|B)_{\rho_2}| \le \eps \log(d_A^2 - 1) + h(\eps) \, .
    \end{equation}
\end{theorem}
\begin{proof}
    Note that since $\rho_{1, B} = \rho_{2, B}$, we further get that $0 = \Tr_A(\rho_1 - \rho_2) = \Tr_A(\eps(\rho_+ - \rho_-))$. This immediately gives $\rho_{+, B} = \rho_{-, B}\equiv \tilde{\rho}_B$ and hence also $S(\rho_{-,B}) = S(\rho_{+, B})$. Then the proof is a simple application of Theorem~\ref{Thm1} and Lemma~\ref{lem:auxiliary-estimate}. Without loss of generality, we assume that $ S(A|B)_{\rho_1} > S(A|B)_{\rho_2} $. Then,
    \begin{equation}\label{eq:estimate-conditional-entropy}
        \begin{aligned}
            S(A|B)_{\rho_1} - S(A|B)_{\rho_2} 
            & = S(\rho_1) - S(\rho_2) \\
            & \le \eps(S(\rho_+) - S(\rho_-)) + h(\eps) \\
            & = \eps(S(A|B)_{\rho_+} - S(A|B)_{\rho_-}) \\
            & \quad + h(\eps). 
        \end{aligned}
    \end{equation}
    To find an upper bound on the difference $\left(S(A|B)_{\rho_+} - S(A|B)_{\rho_-}\right)$  we first express it as a difference of two relative entropies using~(\ref{condrel})
    \begin{align*}
       S(A|B)_{\rho_+} &- S(A|B)_{\rho_-} \nonumber \\
      = & D(\rho_- \Vert \1_A \otimes \tilde{\rho}_B)  - D(\rho_+ \Vert \1_A \otimes \tilde{\rho}_B)\nonumber\\
      = & D(\rho_- \Vert \tau_A \otimes \tilde{\rho}_B)  - D(\rho_+ \Vert \tau_A \otimes \tilde{\rho}_B)
    \end{align*}
    where $\tau_A = \1_A/d_A$ (the completely mixed state). In the above, the last line follows from the fact that for any positive constant $c$, $D(\rho\|c\sigma) = D(\rho \Vert \sigma) - \log c$.
    
    We can complete the proof of the theorem by applying Lemma~\ref{lem:auxiliary-estimate} with the following choices: { $\rho = \rho_-$, $\sigma=\rho_+$, $\omega = \tau_A \otimes \tilde{\rho}_B$ and $t = \frac{1}{d_A^2}$. We can do this because the conditions of the lemma are satisfied for these choices: first, $\rho_+ \perp \rho_-$, and  second, $\rho_- \le d_A^2\tau_A \otimes \tilde{\rho}_{B}$, (which follows from e.g.~\cite[Lemma 5.11]{Tomamichel.2016}) since $\tilde{\rho}_B = \Tr_A (\rho_-)$, and hence we can write
    \begin{equation*}
        \tau_A \otimes \tilde{\rho}_{B} = \frac{1}{d_A^2}  \rho_- + \Big(1 - \frac{1}{d_A^2}\Big) \nu \, ,
    \end{equation*}
    where $\nu = \frac{d_A^2}{d_A^2 - 1}(\tau_A \otimes \tilde{\rho}_{B} - \frac{1}{d_A^2}\rho_-)$.}

    Thus, applying Lemma~\ref{lem:auxiliary-estimate} with the above choices, we obtain
    \begin{equation}
        S(A|B)_{\rho_1} - S(A|B)_{\rho_2} \leq \varepsilon \log (d_A^2 - 1) + h(\eps) \, . 
    \end{equation}
    We finalize the proof by swapping the roles of $\rho_1$ and $\rho_2$, which gives us the absolute value on the left-hand side of he statement.
\end{proof}

\medskip

The search for the tight bound (\ref{eq:conjecture}) conjectured in \cite{Winter.2016, Wilde.2020}, 
in the setting beyond the one in which the marginals on the conditioning system are identical, remains open. 

\medskip

As mentioned earlier, our fundamental inequality, (\ref{ineq1}) of Theorem~\ref{Thm1}, and Lemma~\ref{lem:auxiliary-estimate} can also be used to obtain a continuity bound on the quantum (Umegaki) relative entropy~\cite{Umegaki.1962} which improves upon the best-known bounds \cite{Bluhm.2023.1,Bluhm.2023.3,Audenaert.2011}.  We begin with the particular case in which the second states in both relative entropies are identical.

\begin{theorem}\label{ThmContBoundRelEnt1Input}
     Let $\rho_1, \rho_2 \in \cD(\cH)$, $\sigma \in \cD_+(\cH)$ with $\frac{1}{2}\Vert \rho_1 - \rho_2\Vert_1 = \eps$  where $\eps \in [0, 1]$. Then
    \begin{align}\label{REcont}
       & |  D(\rho_1 \Vert \sigma) - D(\rho_2 \Vert \sigma) | \nonumber \\
       & \le \eps \log( e^{\max\{D_{\max}(\rho_+ \Vert \sigma), D_{\max}(\rho_-\Vert \sigma)\}} - 1) + h(\eps) \, ,
    \end{align}
    where
    $ \rho_1 - \rho_2 = \eps \rho_+ - \eps\rho_- \,,$
   and $\rho_\pm$ is defined via the Jordan-Hahn decomposition of $(\rho_1-\rho_2)$ as in (\ref{two}). 
    
    In particular, for $\rho \in \cD(\cH)$ and $\sigma \in \cD_+(\cH)$ with $\frac{1}{2}\Vert \rho - \sigma\Vert_1= \eps$ where $\eps \in [0, 1]$, we obtain
    \begin{equation}\label{REcontDB}
        D(\rho \Vert \sigma)  \le \eps \log( e^{D_{\text{max}}(\omega_+ \Vert \sigma) } - 1) + h(\eps) \, , 
    \end{equation}
    where $\omega_+$ is defined via the Jordan-Hahn decomposition of $(\rho-\sigma)$ similarly as  $\rho_+$  for $(\rho_1-\rho_2)$.
\end{theorem}

\begin{proof}
    First, note that 
    \begin{equation}\label{eq:bound_diff_relents}
        \begin{aligned}
            & D(\rho_1 \Vert \sigma) - D(\rho_2 \Vert \sigma) \\
            = & - S(\rho_1) + S(\rho_2) - \Tr( ( \rho_1 -\rho_2) \log \sigma) \\
            = & - S(\rho_1) + S(\rho_2) - \varepsilon \Tr((\rho_+ - \rho_-) \log \sigma ) \\
            \leq & \varepsilon( S(\rho_-) - S(\rho_+)) - \varepsilon \Tr((\rho_+ - \rho_-) \log \sigma ) + h(\varepsilon) ,
        \end{aligned}
    \end{equation}
    where in the last line we have used~\eqref{Inineq2}. Thus, rewriting the RHS of~\eqref{eq:bound_diff_relents} as relative entropies, we obtain
    \begin{equation}\label{eq:intermediate_contboundrelent1input}
        D(\rho_1 \Vert \sigma) - D(\rho_2 \Vert \sigma) \leq \varepsilon ( D(\rho_+ \Vert \sigma) -  D(\rho_- \Vert \sigma)) + h(\varepsilon) \, .
    \end{equation}
    Now, note that, from the definition~(\ref{dmax}) of $D_{\text{max}}$, we have
    \begin{equation}
        \rho_+ \leq e^{D_{\text{max}}(\rho_+ \Vert \sigma) } \sigma \, , 
    \end{equation}
    and thus
    \begin{equation}
        \sigma = e^{-D_{\text{max}}(\rho_+ \Vert \sigma) } \rho_+ + \left(1 - e^{-D_{\text{max}}(\rho_+ \Vert \sigma) }\right) \nu 
    \end{equation}
    for a certain $\nu \in \mathcal{D}(\mathcal{H})$. Then, as an immediate consequence of Lemma~\ref{lem:auxiliary-estimate}, we have
    \begin{equation}\label{eq:intermediate_contboundrelent1input2}
        D(\rho_+ \Vert \sigma) -  D(\rho_- \Vert \sigma) \leq \log \left( e^{D_{\text{max}}(\rho_+ \Vert \sigma) } - 1  \right) \, ,
    \end{equation}
    which jointly with~\eqref{eq:intermediate_contboundrelent1input}, and the analogous inequality obtained by swapping the roles of $\rho_1$ and $\rho_2$, allows us to conclude~\eqref{REcont}.  
    
    For~\eqref{REcontDB}, we consider $\rho=\rho_1$ and $\sigma=\rho_2$ in~\eqref{REcont}. Noticing the trivial fact that $D(\rho \Vert \sigma) \ge D(\sigma \Vert \sigma)$ immediately yields the desired inequality.
\end{proof}

The uniform continuity bound for the relative entropy in the first argument provided in~\eqref{REcont} can be compared with the findings of~\cite[Eq. (43) and (44)]{Gour.2020} (also based on the previous work~\cite{Gour.2021}), where it was shown that 
\begin{equation} \label{REcont1inputTG}
    |D(\rho_1 \Vert \sigma) - D(\rho_2 \Vert \sigma)| \le \max\limits_{i = 1, 2} \, \log \left( 1 + \frac{\norm{\rho_1 - \rho_2}_\infty}{\lambda_{\min}(\rho_i) \lambda_{\min}(\sigma)} \right) \, ,
\end{equation}
whenever $\rho_1, \rho_2 \in \cD_+(\cH)$ and $\min\limits_{i = 1, 2} \lambda_{\min}(\rho_i) > \norm{\rho_1 - \rho_2}_\infty$. Additionally,~\eqref{REcont} can be compared with 
\begin{equation}\label{REcont1inputBCGP}
    |D(\rho_1 \Vert \sigma) - D(\rho_2 \Vert \sigma)| \le \eps \log \lambda_{\min}(\sigma)^{-1} + (1 + \eps) h \Big(\frac{\eps}{1 + \eps}\Big) \, , 
\end{equation}
from~\cite[Corollary 5.9]{Bluhm.2023.1}. The comparison between~\eqref{REcont},~\eqref{REcont1inputTG} and~\eqref{REcont1inputBCGP} is made explicit in Figure~\ref{fig:plot}, where it is clear that our new bound~\eqref{REcont} outperforms the others.

\begin{figure}[ht!]
    \centering
    \includegraphics[scale=0.4]{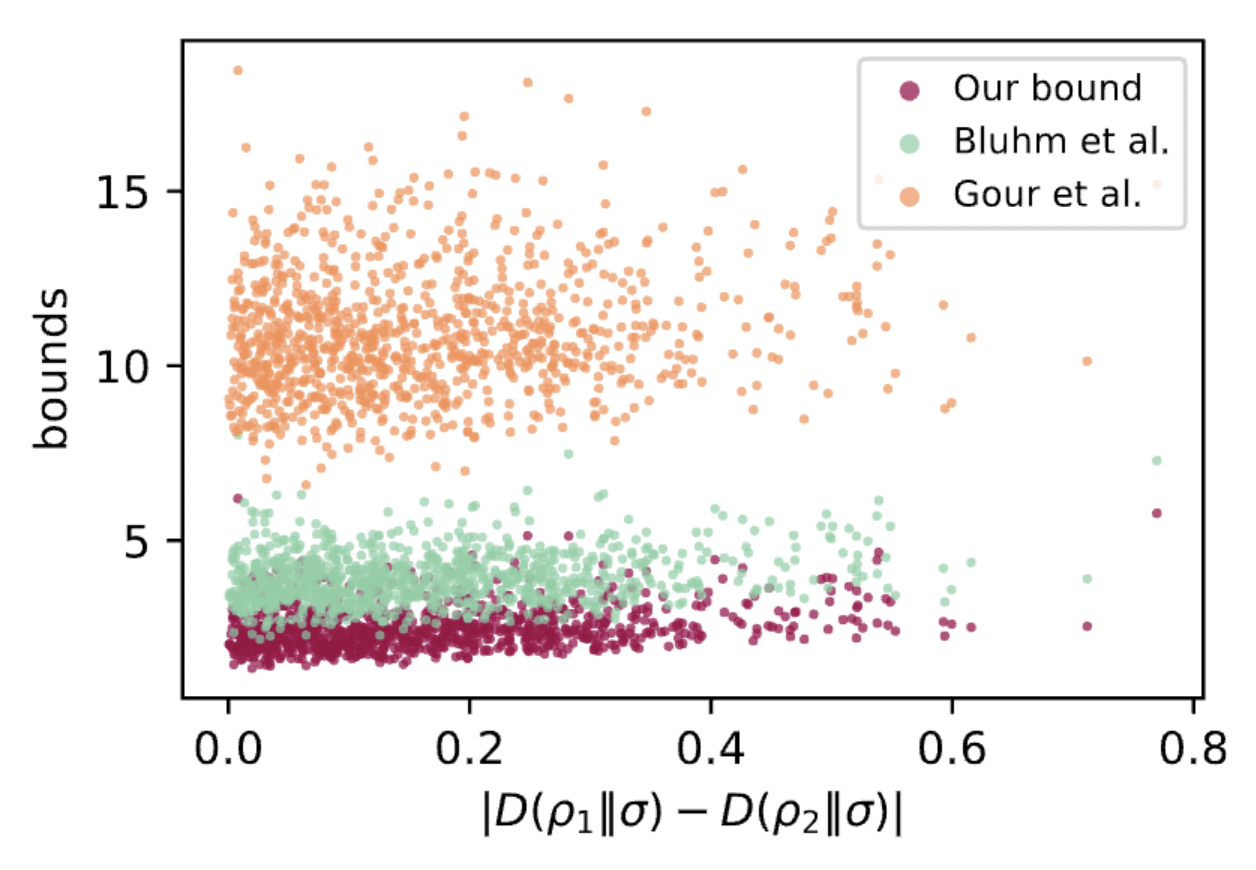}
    \caption{Comparison between the bounds provided in~\eqref{REcont} [our bound],~\eqref{REcont1inputTG} [Gour et al.] and~\eqref{REcont1inputBCGP} [Bluhm et al.]. Here, we have taken $d=15$, and we have randomly generated $1000$ triples of density matrices $\rho_1$, $\rho_2$ and $\sigma$, for which we have plotted the corresponding values of the three bounds considered. This shows that~\eqref{REcont}  is always slightly better than~\eqref{REcont1inputBCGP}, and better than~\eqref{REcont1inputTG}.}
    \label{fig:plot}
\end{figure}

Next, we can easily extend Theorem \ref{ThmContBoundRelEnt1Input} to a continuity bound in both arguments of the relative entropy in the following way.

\begin{corollary}\label{cor:RelEntbothinputs}
       Let $\rho_1, \rho_2 \in \cD(\cH)$ $\sigma_1, \sigma_2 \in \cD_+(\cH)$ with $\frac{1}{2}\Vert \rho_1 - \rho_2 \Vert_1= \eps$ and $\frac{1}{2} \Vert \sigma_1 -\sigma_2\Vert_1 \leq \delta$ where $\eps, \delta \in (0, 1)$. Then
       \begin{equation}\label{REcontbothinputs}
       \begin{aligned}
       & |  D(\rho_1 \Vert \sigma_1) - D(\rho_2 \Vert \sigma_2) | \\
       & \le \eps \log( e^{\max\{D_{\max}(\rho_+ \Vert \sigma_1), D_{\max}(\rho_-\Vert \sigma_2)\}} - 1) \\
       & \quad + \log(1 + \delta \lambda_{\min}(\sigma)^{-1}) + h(\eps) \, ,
       \end{aligned}
       \end{equation}
       where $ \rho_1 - \rho_2 = \eps \rho_+ - \eps\rho_- \,,$
       and $\rho_\pm$ is defined via the Jordan-Hahn decomposition of $(\rho_1-\rho_2)$ as in (\ref{two}), and $\lambda_{\min}(\sigma) := \min\{\lambda_{\min}(\sigma_1), \lambda_{\min}(\sigma_2)\}$.
\end{corollary}

\begin{proof}
Let us first begin by rewriting the difference of relative entropies in terms of another difference of relative entropies for which the second states are identical:
       \begin{equation}\label{diff}
        \begin{aligned}
            D(\rho_1 \Vert \sigma_1) - D(\rho_2 \Vert \sigma_2) 
            & = D(\rho_1 \Vert \sigma_1) - D(\rho_2 \Vert \sigma_1) \\
            & \quad + \Tr ( \rho_2 (\log \sigma_2 - \log \sigma_1)) .
        \end{aligned}
    \end{equation}
    For the first two terms above, we use \eqref{eq:intermediate_contboundrelent1input} and \eqref{eq:intermediate_contboundrelent1input2} from the proof of \ref{ThmContBoundRelEnt1Input}, obtaining the first and last terms of the RHS of \eqref{REcontbothinputs}. For the last term in \eqref{diff}, note that $\sigma_2 \leq e^{D_{\max}(\sigma_2 \Vert \sigma_1)} \sigma_1$. Hence, $\log \sigma_2 \leq D_{\max} (\sigma_2 \Vert \sigma_1)\II + \log \sigma_1$. This in turn implies that 
    \begin{equation}
        \log \sigma_2 - \log \sigma_1 \leq D_{\max}(\sigma_2 \Vert \sigma_1) \II
        = \II \log \Vert \sigma_1^{-1/2} \sigma_2 \sigma_1^{-1/2} \Vert_\infty.
    \end{equation}
    From this we get
    \begin{equation}\label{tt3}
        \begin{aligned}
              \Tr \left( \rho_2 (\log \sigma_2 - \log \sigma_1) \right)
            & \leq (\Tr \rho_2) D_{\max}(\sigma_2 \Vert \sigma_1) \\
            & = \log \Vert \sigma_1^{-1/2} \sigma_2 \sigma_1^{-1/2} \Vert_\infty \, .
        \end{aligned}
    \end{equation}
    Since we can write $\sigma_2$ as
    \begin{equation}
        \sigma_2 = \sigma_1^{1/2} \left(  \sigma_1^{-1/2} (\sigma_2 - \sigma_1)  \sigma_1^{-1/2} + \II \right) \sigma_1^{1/2}  \, ,
    \end{equation}
    we have from (\ref{tt3})
    \begin{equation}\label{fint3}
        \begin{aligned}
            & \Tr \left( \rho_2 (\log \sigma_2 - \log \sigma_1) \right) \\
            & \quad \quad \leq \log \Vert \sigma_1^{-1/2} (\sigma_2 - \sigma_1)  \sigma_1^{-1/2} + \II  \Vert_\infty \\ 
            & \quad \quad \leq \log \left(1 + \Vert \sigma_1^{-1}\Vert_\infty \Vert \sigma_2 - \sigma_1 \Vert_\infty \right)\\
            & \quad \quad \leq  \log \left( 1 + \delta \lambda_{\min}^{-1}(\sigma) \right) \, ,
        \end{aligned}
    \end{equation}
    where we have used the fact that $\Vert \sigma_2 - \sigma_1 \Vert_\infty \leq (1/2) \Vert \sigma_2 - \sigma_1 \Vert_1 = \delta$ which can be deduced from the Jordan-Hahn decomposition.
\end{proof}

To the best of our knowledge, the only previously existing bound in the form of~\eqref{REcontbothinputs} is that of~\cite[Theorem 5.13]{Bluhm.2023.1}, where it was shown that
\begin{align}\label{REcontBCGP} 
    \begin{aligned}
        & |D(\rho_1 \Vert \sigma_1) - D(\rho_2 \Vert \sigma_2)| \\
        &\leq  \left( \eps +  \frac{3 \delta}{1- \frac{\lambda_{\min} (\sigma)}{2}}   \right)  \log \left( {2}{\lambda_{\min} (\sigma)}^{-1} \right) \\
        & \quad + (1+\eps) h \left( \frac{\eps}{1+\eps} \right) + 2 \log\left(1 + \frac{ {2}{\lambda_{\min} (\sigma)}^{-1}  \delta }{ 1- \frac{\lambda_{\min} (\sigma)}{2} + \delta} \right)\, . 
    \end{aligned}
\end{align}
The RHS of inequality \eqref{REcont} gives a more manageable estimate that outperforms~\eqref{REcontBCGP}.

\begin{remark}\label{RemNonMonotonicGlobal}
   Note that the right hand sides of \eqref{eq:continuity-conditional-entropy} of  Theorem \ref{thm:continuity-conditional-entropy}, (\ref{REcont}) of Theorem~\ref{ThmContBoundRelEnt1Input}, and (\ref{REcontbothinputs}) of Corollary~\ref{cor:RelEntbothinputs}, respectively, are generally not monotonically increasing in $\eps$, and hence it is important to note that we assumed in all cases that the trace distance  $\frac{1}{2} \norm{\rho_1 - \rho_2}_1$ is exactly equal to $\eps$ and not just upper bounded by it. If one wants to obtain a result where one only assumes an upper bound $\frac{1}{2} \norm{\rho_1 - \rho_2}_1 \leq \eps'$, one will have to take a supremum of the RHS of \eqref{REcont} over all $\eps \leq \eps'.$ It is easy to see that the RHS of \eqref{REcont} is monotonically increasing until it hits its maximum at $\eps = 1 - 1/d_A^2$, and hence one then obtains
    \begin{equation}
    \begin{aligned}
        &|S(A|B)_{\rho_1} - S(A|B)_{\rho_2}| \\
        & \le \begin{cases} \eps' \log( d_A^2 - 1) + h(\eps'), & \eps' < 1 - 1/d_A^2 \\
       \log(d_A^2), & \eps' \geq 1-1/d_A^2
       \end{cases}\, ,
    \end{aligned}
    \end{equation}
    which is similar to the result obtained in \cite[Theorem 5]{Berta.2024}. Analogously, for Theorem~\ref{ThmContBoundRelEnt1Input} one gets
    \begin{equation}\label{eq:continuity_bound_relent_firstinput_monotonic}
    \begin{aligned}
        &|  D(\rho_1 \Vert \sigma) - D(\rho_2 \Vert \sigma) | \\
        & \le \begin{cases} \eps' \log( M - 1) + h(\eps'), & \eps' < 1 - 1/M \\
       \log(M), & \eps' \geq 1-1/M
       \end{cases}\, ,
    \end{aligned}
    \end{equation}
     with $M \coloneqq e^{\max\{D_{\max}(\rho_+ \Vert \sigma), D_{\max}(\rho_-\Vert \sigma)\}}$, which is similar to the result obtained in \cite[Theorem 1]{Berta.2024}. Finally, for Corollary~\ref{cor:RelEntbothinputs}, note that this problem is not present in the factor involving $\sigma_1$ and $\sigma_2$, for which we only need to assume $\frac{1}{2} \norm{\sigma_1 - \sigma_2}_1 \leq \delta$. Therefore, as a consequence of this and \eqref{eq:continuity_bound_relent_firstinput_monotonic}, one gets
    \begin{equation}
    \begin{aligned}
        & |  D(\rho_1 \Vert \sigma_1) - D(\rho_2 \Vert \sigma_2) | \\
        & \le \left\lbrace\begin{aligned} & \eps' \log( M - 1) + \log(1 + \delta \lambda_{\min}(\sigma)^{-1})+ h(\eps'), \\
        & \hspace{5.5cm} \eps' < 1 - 1/M \\
       & \log(M) + \log(1 + \delta \lambda_{\min}(\sigma)^{-1}), \\
       & \hspace{5.5cm} \eps' \geq 1-1/M
       \end{aligned}\right. \, .
    \end{aligned}
    \end{equation}
\end{remark}

\vspace{0.2cm}

\section{Extension of the fundamental inequality (Theorem~\ref{Thm1}) to the infinite-dimensional setting}
\label{sec:infinite}

The fundamental inequality (Theorem~\ref{Thm1}) is also valid in the infinite-dimensional setting in following form\footnote{We are very grateful to Maksim Shirokov for pointing this out to us.}: 
\begin{equation*}
    S(\rho_1)+\eps S(\rho_-)\leq S(\rho_2)+\eps S(\rho_+)+h(\eps) \, ,
\end{equation*} 
where one or both sides may equal $+\infty$. This can be derived from the finite-dimensional version by approximation.   

We will use the  homogeneous extension of the von Neumann entropy $S(\rho)=\Tr (\eta(\rho))$ (where $\eta(x)=-x \log x$) to the positive cone of trace class operators $\cT_+(\cH)$ on a separable Hilbert space $\cH$, defined as
\begin{equation}\label{S-ext}
    S(\rho) := (\Tr\rho)S(\rho/\Tr\rho)  = \Tr(\eta(\rho))-\eta(\Tr\rho)
\end{equation}
for any nonzero operator $\rho$ in $\cT_+(\cH)$ and equal to $0$ at the zero operator \cite{Lindblad.1974}.

It is easy to see (cf., f.i., \cite[p.1541]{Shirokov.2011}) that
\begin{equation}\label{0}
    S(c\rho)=c S(\rho),\quad c\geq 0,
\end{equation}
and
\begin{equation}\label{1}
    S(\rho+\sigma)\leq S(\rho)+S(\sigma)+h(\Tr\rho,\Tr\sigma),
\end{equation}
for any $\rho$ and $\sigma$ in $\cT_+(\cH)$, where  $h(\Tr\rho,\Tr\sigma)=\eta(\Tr\rho)+\eta(\Tr\sigma)-\eta(\Tr(\rho+\sigma))$
is the homogeneous extension of the binary entropy to the positive cone in $\mathbb{R}^2$.
\smallskip

\noindent
\textbf{Note:} An equality holds in (\ref{1}) if and only if $\rho\sigma=0$.
\smallskip

Assume that $\rho_1, \rho_2$ are arbitrary states in $\cD(\cH)$ (where in infinite dimensions \(\cD(\cH)\) denotes the positive trace class operators with unit trace) and fix $\rho_\pm$ to be the unique states defined through the Jordan-Hahn decomposition $\rho_1 - \rho_2 = \Delta_+ - \Delta_- = \varepsilon (\rho_+ - \rho_-)$, where $\eps = \frac{1}{2}\norm{\rho_1-\rho_2}_1\neq0$. We want to show that 
\begin{equation}\label{m-in}
    S(\rho_1) + \eps S(\rho_-)\leq S(\rho_2)+\eps S(\rho_+) + h(\varepsilon)
\end{equation}
(where one or both sides may be equal to $+\infty$) using our Theorem~\ref{Thm1} valid for finite rank states.
\smallskip

Let $\{P_n\}$ and $\{Q_n\}$ be non-decreasing sequences of spectral projectors 
of the operators $\Delta_+$ and $\Delta_-$ strongly converging, respectively, 
to the projectors $P_*$ and $Q_*$ onto the supports of $\Delta_+$ and $\Delta_-$. Let further  $R=I_{\cH}-P_*-Q_*$.

Let $\cH'$ be any finite-dimensional space and $\omega$ be a given pure state in $\cS(\cH')$.
Consider the sequences of states  
\begin{equation*}
    \hat{\rho}_{1,n} = (P_n+Q_n+R)\rho_1 (P_n+Q_n+R)\oplus p_n\omega,
\end{equation*}
and
\begin{equation*}
    \hat{\rho}_{2, n} =(P_n+Q_n+R)\rho_2 (P_n+Q_n+R)\oplus q_n\omega,
\end{equation*}
in $\cS(\cH\oplus \cH')$, where $\,p_n=1-\Tr ((P_n+Q_n+R)\rho_1)\,$ and $\,q_n=1-\Tr ((P_n+Q_n+R)\rho_2)$. Employing Jordan-Hahn again we decompose $\hat{\rho}_{1, n} - \hat{\rho}_{2, n}$ into $\hat{\Delta}_{\pm, n}$ and note that
\begin{equation*}
    \hat{\Delta}_{+, n} = P_n\Delta_+ \oplus[p_n-q_n]_+\omega,
\end{equation*}
and
\begin{equation*}
    \quad \hat{\Delta}_{-, n} = Q_n\Delta_-\oplus[p_n-q_n]_-\omega.
\end{equation*}
($[x]_+=\max\{x,0\}$, $[x]_-=\max\{-x,0\}$). Indeed, it is clear that the positive operators in the RHS of the above expressions have orthogonal supports. Note also that 
the difference between these operators is equal to  
\begin{align*}
    (P_n + &Q_n + R)(\Delta_+ - \Delta_-)(P_n+Q_n+R) \\
    & \qquad \oplus([p_n-q_n]_+-[p_n-q_n]_-)\omega\\
    &=(P_n+Q_n+R)(\rho_1 - \rho_2)(P_n+Q_n+R) \\
    & \qquad \oplus(p_n-q_n)\omega \\
    & = \hat{\rho}_{1, n}-\hat{\rho}_{2, n} \, .
\end{align*}
Let $\hat{\rho}_{\pm, n} = \eps_n^{-1}\hat{\Delta}_{\pm, n}$ be states in $\cS(\cH\oplus \cH')$, where  
\begin{equation*}
\begin{aligned}
    \eps_n & = \frac{1}{2}\norm{\hat{\rho}_{1, n} - \hat{\rho}_{2, n}}_1 \\
    & = \frac{1}{2}\left(\norm{(P_n+Q_n+R)(\rho_1 - \rho_2)(P_n+Q_n+R)}_1 \right. \\
    &\qquad \quad  \left. +|p_n-q_n|\right).
\end{aligned}
\end{equation*}
It follows from (\ref{0}), (\ref{1}) and the remark after (\ref{1}) that  
\begin{align*}
    S(\hat{\rho}_{1, n}) & = S((P_n+Q_n+R)\rho_1 (P_n+Q_n+R)) \\
    & \qquad + h(p_n,(1-p_n)) \, ,\\
    S(\hat{\rho}_{2, n}) & = S((P_n+Q_n+R)\rho_2 (P_n+Q_n+R)) \\
    & \qquad + h(q_n,(1-q_n)) \, ,
\end{align*}
and that 
\begin{align*}
    S(\hat{\rho}_{+, n}) 
    &= \eps^{-1}_n(S(P_n\Delta_+)+h(c_n,[p_n-q_n]_+)) \, , \\
    S(\hat{\rho}_{-, n})
    &= \eps^{-1}_n (S(Q_n\Delta_-)+h(d_n,[p_n-q_n]_-)) \, ,
\end{align*}
where
\begin{align*}
    & c_n = \Tr (P_n\Delta_+) \, ,\\
    & d_n = \Tr (Q_n\Delta_- )\, .
\end{align*}
Since $\hat{\rho}_{1, n}$ and $\hat{\rho}_{2, n}$ are finite-rank states for each $n$, Theorem \ref{Thm1} implies that
\begin{equation}\label{m-in-n}
    S(\hat{\rho}_{1, n}) + \eps_n S(\hat{\rho}_{-, n})\leq S(\hat{\rho}_{2, n}) + \eps_n S(\hat{\rho}_{+, n}) + h(\eps_n) \quad \forall n. 
\end{equation}
By Lemma 4 in \cite{Lindblad.1974} we have
\begin{equation*}
    \lim_{n\to+\infty} S((P_n+Q_n+R)\rho_1(P_n+Q_n+R)) = S(\rho_1)\leq+\infty
\end{equation*}
and
\begin{equation*}
    \lim_{n\to+\infty} S((P_n+Q_n+R)\rho_2(P_n+Q_n+R)) = S(\rho_2) \leq +\infty \, .
\end{equation*}
It is also clear that 
\begin{equation*}
    \lim_{n\to+\infty} S(P_n\Delta_+) = S(P_*\Delta_+) = S(\Delta_+) \leq +\infty
\end{equation*}
and
\begin{equation*}
    \lim_{n\to+\infty} S(Q_n\Delta_-)=S(Q_*\Delta_-)=S(\Delta_-) \leq +\infty \, .
\end{equation*}
Since $\eps_n$ tends to $\eps$, $p_n,q_n$ tend to $0$ and $c_n,d_n$ tend to $1$, using these limits relations and the expressions before (\ref{m-in-n}) one can prove (\ref{m-in}) by taking the limit in (\ref{m-in-n}) as $n\to+\infty$.

\section{Strengthened Inequalities}
Note that for all our applications in Sections \ref{sec:ref-AF} and \ref{sec:applications} we relied on the inequality (\ref{ineq1}) of Theorem~\ref{Thm1} instead of the sharper inequality given by (\ref{ineq1-sharp}) of Theorem~\ref{Thm-sharp}. This is because (\ref{ineq1}) is simpler.
However, the entropic inequalities of Sections \ref{sec:ref-AF} and \ref{sec:applications} can be easily strengthened by making use of Theorem~\ref{Thm-sharp} instead of Theorem~\ref{Thm1}. For example, the strengthened forms of (\ref{ineq1p}) of Theorem~\ref{thm-ref-AF}, (\ref{eq:continuity-conditional-entropy}) of Theorem~\ref{thm:continuity-conditional-entropy} and (\ref{REcont}) of Theorem~\ref{ThmContBoundRelEnt1Input}, are respectively given by
\begin{align*}
        &|S(\rho_1) - S(\rho_2)|  \\ & \qquad \leq \eps \log\big(d \max\{\lambda_{\max}(\rho_-), \lambda_{\max}(\rho_+)\} - 1\big) + c \, h\left(\frac{\eps}{c}\right) ; \\&|S(A|B)_{\rho_1} - S(A|B)_{\rho_2}| \le \eps \log(d_A^2 - 1) + c \, h\left(\frac{\eps}{c}\right) ;\\
        & |  D(\rho_1 \Vert \sigma) - D(\rho_2 \Vert \sigma) | \nonumber \\
       & \qquad \le \eps \log( e^{\max\{D_{\max}(\rho_+ \Vert \sigma), D_{\max}(\rho_-\Vert \sigma)\}} - 1) + c \, h\left(\frac{\eps}{c}\right), \, 
\end{align*}
where $c = \Tr(\rho_2 |_{\supp \rho_-})$ in all three cases.

\bigskip

\noindent \textbf{Note:} Our result, in the finite-dimensional setting, on the uniform continuity bound for the conditional entropy in the case of equal marginals (Theorem \ref{thm:continuity-conditional-entropy}) was obtained with different techniques by Berta, Lami, and Tomamichel~\cite{Berta.2024}.

\section*{Acknowledgments}

ND would like to thank Mohammad Alhejji for igniting her curiosity in the bound~(\ref{eq:conjecture}). She is grateful to Milan Mosonyi for the invitation to the {\em{Focused Workshop on R\'enyi Divergences}} at the Erd\"os Center in Budapest, in July 2024, during which the proof of Theorem~\ref{Thm1} was completed. It was also at that workshop that Peter Frenkel raised a question which led us to prove Theorem~\ref{Thm-sharp}, which is a sharper version of Theorem~\ref{Thm1}, so we would like to thank him too. ND would also like to thank Ludovico Lami, Mario Berta and Marco Tomamichel for helpful discussions in Budapest and Cambridge, where they shared their preliminary results with her, which motivated her and her coauthors  to prove Theorems~\ref{thm:continuity-conditional-entropy} and~\ref{ThmContBoundRelEnt1Input}. BB acknowledges support from the UK Engineering and Physical Sciences Research Council (EPSRC) under grant number EP/V52024X/1.
MGJ acknowledges support from the Fonds de la Recherche Scientifique – FNRS (Belgium) and by the Wiener-Anspach Foundation (Belgium), and is grateful for the hospitality of the Centre for Mathematical Sciences at the University of Cambridge where this work was performed during an extended research stay.
AC and PG acknowledge support from the Deutsche Forschungsgemeinschaft (DFG, German Research Foundation) – Project-ID 470903074 – TRR 352. For the purpose of open access, the authors have applied a Creative Commons Attribution (CC BY) license to any Accepted Manuscript version arising. The authors would like to thank Maksim Shirokov for the proof of the extension of our fundamental inequality (Theorem~\ref{Thm1}) to the infinite-dimensional setting (see Section~\ref{sec:infinite}).

\bibliographystyle{IEEEtran}
\bibliography{references}

\end{document}